\documentclass[11pt]{amsart}
\usepackage{macros,slashed,amsrefs}

\newtheorem{thm}{Theorem}[section]
\newtheorem{lem}[thm]{Lemma}
\newtheorem{prop}[thm]{Proposition}
\newtheorem{cor}[thm]{Corollary}
\newtheorem*{thma}{Theorem A}
\newtheorem*{thmb}{Theorem B}

\theoremstyle{definition}
\newtheorem{dfn}[thm]{Definition}

\theoremstyle{remark}
\newtheorem{rmk}[thm]{Remark}

\numberwithin{equation}{section}

\begin{document}

\title[RG flow and the $\sigma$ model]{Homotopy RG flow and the non-linear $\sigma$-model}

\author{Ryan Grady}
\address{Montana State University\\Bozeman 59717\\USA}
\email{ryan.grady1@montana.edu} 
\thanks{R Grady was partially supported by the National Science Foundation under Award DMS-1309118.}

\author{Brian Williams}
\address{Northwestern University\\ Evanston 60208\\USA}
\email{bwill@math.northwestern.edu}
\thanks{B Williams enjoyed support of the
National Science Foundation as a graduate student research fellow under Award DGE-1324585.}

\subjclass[2010]{Primary 81T40, 81T17, 53C44; Secondary 81T15, 17B55}

\date{}

\begin{abstract}
The purpose of this note is to give a mathematical treatment to the low energy effective theory of the two-dimensional sigma model. Perhaps surprisingly, our low energy effective theory encodes much of the topology and geometry of the target manifold. In particular, we relate the $\beta$-function of our theory to the Ricci curvature of the target, recovering the physical result of Friedan. 

\end{abstract}

\maketitle

\tableofcontents

\section{Introduction}

In this note we give a mathematical treatment to the low energy effective theory of the two-dimensional $\sigma$-model.  
This $\sigma$-model involves maps from a Riemann surface $\Sigma$ to a Riemannian manifold $X$ with metric $h$. 
Physically, a field is a smooth map $\varphi : \Sigma \to X$ and the action functional is given by
\ben
S(\varphi) = \int_{\Sigma} h(\partial \varphi, \Bar{\partial} \varphi),
\een
where $\partial, \Bar{\partial}$ are the holomorphic and anti-holomorphic pieces of the de Rham differential on $\Sigma$. 
The map $\varphi$ satisfies the classical equations of motion if and only if it is harmonic. 
Classically, of course, the theory is conformal.
In particular, when we work locally on $\Sigma$ the theory is scale invariant. 
We study the failure of the classical theory to be scale invariant at the quantum level. The quantity that measures this is precisely the {\em $\beta$-function} of the quantum field theory. 

We do not study the full theory, but rather only the theory in perturbation around the space of constant maps in the space of all harmonic maps. 
Our low energy effective quantization encodes much of the topology and geometry of the target manifold. 
Our main result is to relate the one-loop $\beta$-function to the Ricci curvature of $X$. 

\begin{thma}
The one-loop $\beta$-function for the two-dimensional  $\sigma$-model with target the Riemannian manifold $(X,h)$ satisfies
\[
\beta^{(1)} (h) = -\frac{1}{12 \pi} \mathrm{Ric} (h).
\]
\end{thma}

In 1985,  Dan Friedan \cite{Fried} gave a physical argument for the validity of Theorem A; see also the more recent geometric overview of Carfora \cite{Car}. This manuscript is by no means the first treatment of the Riemannian $\sigma$-model using the Batalin-Vilkovisky (BV) formalism. Most recently, using a different approach, Nguyen \cite{Nguyen} has constructed the BV quantization of the theory for a general target manifold and recovers Friedan's result from it.
 Below, we give a mathematically rigorous treatment for a general target Riemannian manifold $(X,h)$ using the language of $L_\infty$ spaces. Our approach also uses Costello's paradigm for BV theory, as well as the subsequent (rigorous) development of the $\beta$-function in this setting by Elliott, Williams, and Yoo \cite{EWY}.

The theorem identifies the one-loop contribution to the $\beta$-function, so to compute it we only consider the quantization to first-order in the quantum parameter $\hbar$. 
Key to this process, in the BV formalism, is to study the moduli of all such quantizations. 
Perturbatively, the object which controls this moduli space is the deformation complex.
Its cohomology represents the space of all functionals by which we can deform the classical action functional. 

\begin{thmb}
Let $\Def$ denote the obstruction deformation complex for the $\sigma$-model of maps $\CC \to (X,h)$.
There is a quasi-isomorphism
\ben
(\Def)^{{\rm Aff}(\CC)} \simeq \Riem(X,h) [1] \oplus \Omega^3_{cl, X},
\een
where $\Omega^3_{cl,X}$ is the sheaf of closed three-forms and
\[
\Riem(X,h) = T_X \xto{\d} \Sym^2(T_X^\vee) [-1].
\]
\end{thmb}

The first summand of $\Def$ deforms the metric on the target manifold, while the second summand introduces an $H$-flux.   In this note, we don't discuss how the addition of an $H$-flux modifies the $\beta$-function.  Though we expect that one could study this through the exact Courant algebroid determined by the $H$-flux 3-form; this would provide an interpretation of recent work of \v{S}evera and Valach \cite{SV} in the BV formalism.

There are many remaining mathematical questions, especially regarding the observable theory for the $\sigma$-model, which are beyond the scope of this paper.  For instance, the algebra of quantum observables \cite{cg2} should be quite large and include the $bc$-$\beta \gamma$-vertex algebra. 
There is also the question of studying the non-perturbative behavior of the $\beta$-function. 
The full $\beta$-function should be a section of some line bundle on the full space of classical solutions (that is, harmonic maps)  obtained by quantizing the theory in families. 

In Section 2, we construct the classical BV theory using Lie theoretic techniques as in \cite{wg2} and \cite{gg1}. We also compute the obstruction deformation complex of the theory, culminating in Proposition \ref{prop:obsdef}. Section 3 is about quantizing the $\sigma$-model to one-loop, i.e., modulo $\hbar^2$; a cohomological argument shows there is no obstruction. In Section 4, we summarize the relevant results from \cite{Cos1} and \cite{EWY} regarding the mathematical definition of the $\beta$-function. In particular, we have aimed to point out connections to and differences from the physics literature, e.g., as in Remark \ref{rmk:warning}.  Finally, in Section 5, we compute the one-loop $\beta$-function in terms of the Ricci tensor of the target manifold.

\subsection{Acknowledgements}

It's a pleasure to thank Kevin Costello, Chris Elliott, Owen Gwilliam, Steve Rosenberg, and Philsang Yoo for helpful feedback and discussion.  
We owe an especially large debt to Si Li and thank him for his insight and suggestions over the course of this project.
We appreciate Eric Berry's careful reading of a preliminary version of this work that helped correct numerous typos. 
In addition, we thank the referee for providing useful feedback that greatly improved the organization and exposition of the paper.

\section{The classical $\sigma$-model}\label{sect:classical}

In this section we set up the mathematical model that we use to study the the $\sigma$-model of maps from a Riemann surface $\Sigma$ to a Riemannian manifold $(X,h)$ where $h$ is the metric. 
Physically, a field is a map $\varphi : \Sigma \to X$ and the action functional is of the form
\ben
S(\varphi) = \int_{\Sigma} h(\partial \varphi, \Bar{\partial} \varphi),
\een
where $\partial, \Bar{\partial}$ are the holomorphic and anti-holomorphic pieces of the de Rham differential on $\Sigma$. 
The classical solutions to the equations of motion are harmonic maps $\Sigma \to X$. 
This whole mapping space is quite complicated and difficult to study, 
and the approach we take is to work perturbatively around a fixed harmonic map $\varphi_0 : \Sigma \to X$ (in fact, we take $\varphi_0$ to be a constant map). 

\subsection{Smooth geometry and $L_\infty$ algebras} 

The language we will use to describe the perturbative non-linear $\sigma$-model uses the formalism of $L_\infty$ spaces.
For examples on how this formalism has been used to study related theories see \cite{gg1, wg2, LiLi, GLL, ggw}.  
The goal is to describe the target by Lie algebraic data so that the theory we write down behaves formally like a gauge theory. 
In later sections we will see that such a description lends itself to a rigorous analysis in the BV-BRST formalism, and hence in a description of its behavior under local RG flow. 

Let $T_X$ be the sheaf of smooth vector fields and $\Omega^*_X$ be the sheaf of de Rham forms equipped with the de Rham differential. 
If we want to refer to the sheaf as a graded vector space, that is with the differential turned off, we use the notation $\Omega^\#_X$. 
If $\sE$ is the sheaf of smooth sections of a vector bundle $E$ with flat connection we let $\Omega^*_X(\sE)$ denote its associated de Rham complex.

The main input we need is a description of $X$ in terms of a certain curved $L_\infty$ algebra defined over the de Rham complex $\Omega^*_X$.  
For a definition of curved $L_\infty$ algebras and a proof of the following result we refer the reader to \cite{gg2}. 

\begin{prop}[Lemma 4.12 in \cite{gg2}] Let $X$ be a smooth manifold. Then, there is a curved $L_\infty$ algebra $\fg_X$ over $\Omega^*_X$ such that
\ben
\clie^*(\fg_X) \cong \Omega^*_X(\sJ_X),
\een
where $\sJ_X$ is the sheaf of $\infty$-jets of the trivial bundle on $X$. 
\end{prop}

\begin{rmk} As a graded $\Omega^\#_X$-module there is an identification $\fg_X \cong \Omega^\#_X \tensor_{C^\infty_X} T_X[-1]$. Moreover, the curved $L_\infty$ algebra $\fg_X$ is a natural object to consider for any smooth manifold $X$: up to isomorphism it is unique up to a contractible choice. 
\end{rmk} 

We will denote by $\{\ell_k\}$ the family of $L_\infty$ brackets of the curved $L_\infty$ algebra $\fg_X$. 
For example, $\ell_0 \in \Omega^*_X$ is of cohomological degree two and represents the ``curving" \footnote{To fix an $L_\infty$ structure on $\fg_X$ it suffices to choose a connection on $T_X$. In this case, $\ell_0$ is identified with the curvature of this connection.}. 
Similarly, the unary bracket $\ell_1 : \fg_X \to \fg_X$ gives $\fg_X$ the structure of a dg module over the dg ring $\Omega^*_X$. 

A choice of a Riemannian metric $h$ on $X$ determines an isomorphism $h : T_X \cong T_X^*$. 
This isomorphism extends to the $L_\infty$ algebra $\fg_X$ in the following way. 

\begin{lem}\label{lem:silly} Suppose $h$ is a Riemannian metric on $X$. 
Then, there is an isomorphism of $\Omega^*_X$-modules
\ben
h^\flat : \fg_X \cong \fg_X^\vee,
\een
where $\fg_X^\vee$ denotes the $\Omega^*_X$-linear dual. 
The inverse isomorphism is denoted $h^\#$. 
\end{lem}
\begin{proof}
Since the isomorphism $h : T_X \cong T_X^\vee$ is of $C^\infty_X$-modules we have an induced isomorphism of graded $\Omega^\#_X$-modules
\ben
h : \Omega^\#_X \tensor_{C^\infty_X} T_X[-1] \cong \Omega^\#_X \tensor_{C^\infty_X} T^\vee_X [-1] .
\een
\end{proof}

Given any vector bundle $E$ on $X$ we can consider the $D_X$-module of smooth $\infty$ jets $\sJ_E$, and hence its de Rham complex $\Omega^*_X$. 
If $\sE$ denotes the sheaf of sections of $E$, there is a natural quasi-isomorphism
\ben
j_\infty : \sE \xto{\simeq} \Omega^*_X(\sJ_E)
\een 
sending a smooth section to its power series expansion. 
For example, in the case of the trivial bundle on $X$ this defines a quasi-isomorphism $C^\infty_X \simeq \clie^*(\fg_X)$. 
In general, we make the following definition. 

\begin{dfn} Let $E$ be a vector bundle on $X$. We define the $\fg_X$ module, which we still denote by $E$, whose Chevalley-Eilenberg complex is the $\clie^*(\fg_X)$ dg module 
\ben
\clie^*(\fg_X ; E) := \Omega^*_X(\sJ_E) .
\een 
\end{dfn} 

In the case of a Riemannian manifold $(X,h)$ we have a metric $h \in \Sym^2(T_X^\vee)$.
In particular, we can consider its $\infty$-jet 
\be\label{hinfty}
h_\infty := j_\infty (h) \in \Omega^*_X\left(\sJ \left(\Sym^2 (T_X^\vee)\right)\right) = \clie^*\left(\fg_X ; \Sym^2(T_X^\vee)\right) .
\ee
By the equality on the right-hand side, we can think of $h_\infty$ as being a functional on $\fg_X$ with values in the sheaf of symmetric $(0,2)$-tensors. 

\subsection{The BV formalism}

To write down a theory on $\Sigma$ in the BV formalism it suffices to prescribe a sheaf of elliptic complexes $\sE$ on $\Sigma$ equipped with a $(-1)$-shifted symplectic form -- the space of fields -- together with a local functional on $\sE$ -- the interaction functional. 
For a more precise definition we refer the reader to Definition 5.4.0.3 in \cite{cg2}. 

The sheaf of fields of the Riemannian $\sigma$-model is the sheaf on $\Sigma$ of $\Omega^\#_X$-modules
\ben 
\sE = C^\infty_X \tensor_\CC \fg_X[1] \oplus \Omega_X^2 \tensor_\CC \fg_X^\vee .
\een
We write the fields as $\varphi \in C^\infty_\Sigma \tensor \fg_X$ and $\psi \in \Omega^2_\Sigma \tensor \fg_X^\vee$. 
The shifted symplectic pairing of degree $(-1)$ is defined by
\ben
\<\varphi, \psi\> = \int_\Sigma \<\varphi, \psi\>_{\fg} \in \Omega^\#_X.
\een

We will write the classical action functional $S : \sE \to \Omega^*_X$ in the form
\ben
S = S_{\rm free} + I
\een
where $S_{\rm free}$ is the free part that is quadratic as a functional on $\sE$.
To define this functional consider the operator
\ben
\partial \Bar{\partial} \tensor h^\flat : C^\infty_\Sigma \tensor \fg_X \to \Omega^2_\Sigma \tensor \fg_X^\vee .
\een 
where $h^\flat$ is interpreted as in Lemma \ref{lem:silly}. 
There are internal differentials on $\fg_X$ (and $\fg_X^\vee$) given by $\ell_1$ (and its linear dual).
We let $Q = \partial \Bar{\partial} \tensor h + \ell_1$. 
Define
\ben
S_{\rm free} (\varphi, \psi) = \int_\Sigma \<\varphi, (\partial \Bar{\partial} \tensor h)(\varphi)\> .
\een
Note that $\<\varphi, (\partial \Bar{\partial} \tensor h) \varphi\> = h(\varphi, \partial \Bar{\partial} \varphi)$ so that this reduces to the familiar kinetic term of the $\sigma$-model.
Further, we remark that the operator $Q$ satisfies $Q^2 = \ell_1^2$ which is, in general, not zero, which reflects the fact that the $L_\infty$ algebra is curved.

To define the interaction we first recall from (\ref{hinfty}) that we can view the metric as determining an element $h_\infty$ which we view as a functional on $\fg_X$. 
We can extend this to a functional on $C^\infty_\Sigma \tensor \fg_X$ simply by multiplying the function component. 
Similarly, we can extend the brackets $\ell_k$ defining the $L_\infty$ structure on $\fg_X$.
The interaction is written as a sum of two terms $I = I^h + I^X$ where
\ben
I^h(\varphi) = \sum_{k \geq 1} \frac{1}{k!} \int_\Sigma h_\infty^{(k)} (\partial \varphi, \Bar{\partial} \varphi; \varphi^{\tensor k})
\een
and 
\ben
I^X (\varphi, \psi) = \sum_{k \geq 2} \frac{1}{(k+1)!} \int_{\Sigma} \<\psi, \ell_k(\varphi^{\tensor k}) )\>_{\fg_X} .
\een
For each $k$, $h_\infty^{(k)}$ denotes the $k$th term in the jet expansion and the `;' separates the tensorial indices from the jet indices.

Thus, the full action functional can be written succinctly as
\ben
S = \int_\Sigma h_\infty(\partial \varphi, \Bar{\partial} \varphi; e^\varphi) + \int_\Sigma \<\psi, \ell(e^\varphi)\>_{\fg_X} .
\een

\subsubsection{Functionals}

To describe functionals on $\sE$ we need to introduce the dual to the space of fields. 
This is defined by
\ben
\sE^\vee = {\rm Hom}_{\Omega^\#_X} (\sE, \Omega^\#_X) .
\een
The space of fields is built from the infinite dimensional space of functions on the Riemann surface $\Sigma$, which has the structure of a topological vector space.
The exact category of topological vector spaces we work with will not be important, but is discussed in depth in the Appendix of \cite{cg1}. 
In the notation above, the space of homomorphisms is taken to be those that are continuous.
With this convention, there is a natural embedding $\sE \hookrightarrow \sE^\vee[-1]$ defined by the shifted symplectic form.

The space of functionals $\sO(\sE)$ is defined by 
\ben
\sO(\sE) = \prod_{n \geq 0} {\rm Hom}_{\Omega^\#_X} (\Sym^n (\sE), \Omega^\#_X) .
\een
There is a subspace of {\em local functionals} $\sO_{\rm loc}(\sE) \subset \sO(\sE)$ consisting of functionals of the form $F(\varphi,\psi) = \int \sL$ where $\sL$ is a Lagrangian density. 
For a precise definition see \cite{cg2}.
The symplectic pairing on the space of fields induces a bracket of degree $+1$ on local functionals that we denote by $\{-,-\}$.

\begin{prop}\label{CME} The function $S = S_{\rm free} + I$ is local. 
Moreover, it satisfies the classical master equation
\ben
Q I + \frac{1}{2} \{I,I\} + F_{\ell_1} = 0,
\een\\
where $F_{\ell_1} \in \Oloc(\sE)$ satisfies $\{F_{\ell_1}, -\} = \ell_1^2$.  
\end{prop}

\begin{proof} 
Due to the nature of the symplectic pairing it is immediate that $\{S_h, S_h\} = 0$. 
To see that $\{S_X,S_h\} = 0$ we observe that $S_h$ is defined using the jet expansion of a globally defined tensor, namely the metric. 
Next, we note that the operator $\{S_X, -\}$ is precisely the Chevalley-Eilenberg differential for the $L_\infty$ algebra $C^\infty(\Sigma) \tensor (\fg_X \ltimes {\rm vol}_\Sigma \cdot \fg_X^\vee[-1])$ with $L_\infty$ structure given by tensoring the algebra of functions with the $L_\infty$ structure on the square-zero extension $\fg_X \ltimes \fg_X^\vee[-1]$. 
Thus $\{S_X,S_X\}(\varphi,\psi) = \<\psi, \ell_1^2(\varphi)\>$ is equivalent to the closure of this curved $L_\infty$ structure. 
Since $\ell_1^2$ represents the curvature of the connection, we are done.
\end{proof}

An immediate corollary is that the operator $\{S,-\} = Q + \{I,-\}$ defines a differential on the space of local functionals (considered as an $\Omega^\sharp_X$-module). 
\begin{dfn}
The obstruction deformation complex is defined by
\ben
\Def = \left(\Oloc(\sE), \{S,-\}\right) = \left(\Oloc(\sE), Q + \{I,-\}\right) .
\een
\end{dfn}

\subsection{Calculation of the obstruction deformation complex}\label{sect:obsdef}

The definition of the obstruction deformation complex is a very natural one from the point of view of the classical field theory: natural deformations of the theory are given by deforming the local action functional. 
A local action functional is specified by two pieces of data: a functional on the space of fields depending only on the infinitesimal jet data of the fields (modulo constants), and a density on the manifold. 
A model for this space of local functionals is given by the following tensor product:
\ben
\Oloc(\sE) = {\rm Dens}_{M} \tensor_{D_M} \sO_{red}(\sJ\sE)
\een
where $D_M$ denotes the sheaf of differential operators on $M$. 
To make sense of this tensor product we recall that ${\rm Dens}_M$ has a natural right $D$-module structure and the infinite jet bundle has a canonical flat connection (and hence a left $D$-module structure).

A priori, the space of all deformations is very large and not very manageable. 
In this section, we restrict ourselves to studying deformations that respect certain symmetries of the $\sigma$-model and compute such deformations in terms of geometric quantities on the target Riemannian manifold. 

We consider the theory with source manifold $\Sigma = \CC$. 
In this case, we see that the theory is acted upon by the group of affine linear transformations ${\rm Aff} (\CC) = \CC^\times \ltimes \CC$ given by rotations and translations. 
This is the symmetry that we wish to impose on the deformations of the model. 
The subcomplex of deformations that are invariant for this symmetry is denoted $(\Def)^{{\rm Aff}(\CC)}$. 

\begin{rmk}
By $(\Def)^{\Aff(\CC)}$ we mean the strict (underived) fixed points of the affine group.
The deformation complex $\Def$ is the global sections of a $\Aff(\CC)$-equivariant vector bundle on $\CC$. 
This $\Aff(\CC)$ representation is induced from a finite dimensional $\CC^\times$-representation.
Since $\CC^\times$ is reductive, its derived and underived fixed points agree, so this operation is reasonable in our context.
\end{rmk}

Already, for the translation invariant subcomplex there is a vast simplification of the space of local functionals. 
Indeed, Lemma 6.7.1 of \cite{Cos1} implies that
\ben
\Oloc(\sE)^{\CC} \simeq \CC \cdot \d^2 z \tensor^{\mathbb{L}}_{\CC[\partial_{z},\partial_{\zbar}]} \sO_{red}(\sJ_0 \sE)
\een
where $\sJ_0 \sE$ denotes the fiber of the infinite jet bundle at $0 \in \CC$. 
This statement is a consequence of the fact that a translation invariant local functional is completely determined by its behavior in a neighborhood of the origin. 

Before stating the main result of this section, we will set up some notation. 
Consider the following sheaf of dg vector spaces on the Riemannian manifold $(X,h)$ 
\ben
T_X \xto{\d} \Sym^2(T_X^\vee) [-1]
\een
where the differential is defined by $\d(X) = L_X h$ where $L$ is the Lie derivative. 
In fact, the Lie bracket on vector fields together with the natural action of vector fields on symmetric tensors endows this sheaf with the structure of a sheaf of dg Lie algebras. 
This sheaf of dg Lie algebras describes the formal neighborhood of $(X,h)$ in the moduli space of all Riemannian structures on $X$. 
We will denote this sheaf by $\Riem(X,h)$. 

\begin{prop}\label{prop:obsdef} Consider the classical BV theory $\sE (\CC)$ placed on the Riemann surface $\Sigma = \CC$. 
There is an equivalence of $\Omega^*_X$-modules  
\ben
(\Def)^{{\rm Aff}(\CC)} \simeq \Omega^*_X(\sJ_F)
\een
where $F$ is the sheaf of dg vector spaces on $X$ given by
\ben
F = \Riem(X,h) [1] \oplus \Omega^3_{cl, X}
\een
where $\Omega^3_{cl,X}$ is the sheaf of closed three-forms. 
\end{prop}

As with the de Rham complex of any infinite jet bundle we observe immediately that there is an equivalence of $C^\infty_X$-modules $(\Def)^{{\rm Aff}(\CC)} \simeq F$. 
In fact, we can describe explicitly a quasi-isomorphism 
\ben
\mathfrak{R} : F \to (\Def)^{{\rm Aff}(\CC)}
\een
as follows. 
Let $j_\infty : F \to \Omega^*_X(\sJ_F)$ denote the quasi-isomorphism given by taking the $\infty$-jet of a section. 
For instance, if $\alpha$ is a section in $\Sym^2(T_X)$ we have $j_\infty(\alpha) \in \Omega^*(X , \sJ \Sym^2(T_X))$. 
Moreover, we know that there is an isomorphism of $\Omega^*_X$-modules $\clie^*(\fg_X ; \Sym^2(T_X^\vee))$. 
We obtain a local functional via the formula
\ben
\mathfrak{R}(\alpha) (\varphi, \psi) = \int j_\infty (\alpha)(\partial \varphi, \overline{\partial} \varphi; e^{\varphi}) \d^2 z .
\een

Geometrically speaking, these deformations and symmetries are apparent from the point of view of the $\sigma$-model.
Indeed, we see all deformations of the Riemannian structure on the target manifold.
The formal neighborhood of $(X,h)$ in the moduli of Riemannian structures is described precisely by the dg Lie algebra ${\rm Riem}(X,h)$. 

For the other component of the deformation complex, note that the degree zero cohomology of the sheaf $\Omega^{3}_{cl,X}$ can be thought of as the collection of H-fluxes. The deformation of the action is realized as follows: every closed 3-form $H$ is locally exact, so choose a 2-form $B$ with $dB=H$. We then have the following functional on maps $\varphi : \Sigma \to X$:
\[
S_H (\varphi) = \int_\Sigma \varphi^\ast B.
\]

\subsubsection{}

It is convenient to introduce the notation $\sL = \sE[-1]$ for the $L_\infty$ algebra obtained by shifting the space of fields up by one. 
In this case, the piece of the differential $\{S_X,-\}$ determines an identification
\ben
\left(\sO_{red}(\sJ_0 \sE), \{S_X,-\}\right) = \cred^*(\sJ_0 \sL) .
\een
To handle the translation invariant piece of the deformation complex, Lemma 6.7.1 of \cite{Cos1}, referenced above, says we must compute the derived tensor product $\CC \cdot \d^2 z \tensor^\LL_{\CC[\partial_z,\partial_{\zbar}]} \cred^*(\sJ_0 \sL)$. 
To compute this we choose a free resolution of the trivial $\CC[\partial_z, \partial_{\zbar}]$ module $\CC \cdot \d^2 z$ given by the complex
\ben
M = \left(\CC[\epsilon, \Bar{\epsilon}, \partial_z, \partial_{\zbar}], \d = \frac{\partial}{\partial \epsilon} \partial_z + \frac{\partial}{\partial \Bar{\epsilon}} \partial_{\zbar}\right) 
\een
where $\epsilon,\Bar{\epsilon}$ have degree $-1$. 
To remind us that we are resolving the trivial module $\CC \cdot \d^2 z$ we will write $M \cdot \d^2 z$. 

We see that the derived tensor product computing the invariant deformation complex reduces to
\be\label{cplx1}
\xymatrix{
\ul{-2} & \ul{-1} & \ul{0} \\
(\epsilon \Bar{\epsilon} \cdot \d^2 z) \cdot \cred^*(\sJ_0 \sL) \ar[r]^-{\d_1} & (\epsilon\cdot \d^2 z) \cdot \cred^*(\sJ_0 \sL) \oplus (\Bar{\epsilon} \cdot \d^2 z)\cdot \cred^*(\sJ_0 \sL) \ar[r]^-{\d_2} & \d^2 z \cdot \cred^*(\sJ_0 \sL) 
}
\ee
where the differentials $\d_1, \d_2$ are determined by the action of $\CC[\partial_z, \partial_{\zbar}]$ on $\sJ_0 \sE$. 
We have omitted the contribution of the part of the differential given by $\{S_h,-\}$. 

Now, we see that the $\infty$-jets at $0 \in \CC$ of the $L_\infty$ algebra $\sL$ are of the form
\ben
J_0 \sL = \fg_X \ltimes \d^2 y \cdot \fg_X^\vee [-1] [y,\Bar{y}],
\een
where we have denoted the formal jet variables by $y$ and $\Bar{y}$.

We have yet to account for invariance for the other piece of the affine group $\CC^\times$ given by scaling. 
Let $\lambda \in \CC^\times$.
The volume element $\d^2 z$ scales as $\lambda \cdot \d^2 z = \lambda \Bar{\lambda} \d^2 z$.
On the variables $\epsilon, \Bar{\epsilon}$ scaling acts as $\lambda \cdot \epsilon = \lambda \epsilon$ and $\lambda \Bar{\epsilon} = \Bar{\lambda} \Bar{\epsilon}$.
Similarly, $\CC^\times$ acts on the jet complex via $\lambda \cdot y = \lambda y$ and $\lambda \cdot \Bar{y} = \Bar{\lambda} \Bar{y}$. 
Thus, in the complex (\ref{cplx1}) above we see that the $\CC^\times$-invariant subcomplex will be given by summands labelled by the following $\CC^\times$-invariant elements:
\ben
\epsilon \Bar{\epsilon} \d^2 z \;\; , \;\; \Bar{\epsilon} \d^2 z (y^\vee) \;\; , \;\; \epsilon \d^2 z (\Bar{y}^\vee) \;\; , \;\; \d^2 z (y^\vee \Bar{y}^\vee) \;\; , \;\; \d^2 z (\d^2 y)^\vee .
\een
We read these terms off as follows:
\begin{itemize}
\item[(1)] the term labeled by $\epsilon \Bar{\epsilon} \d^2 z$ contributes the summand $\cred^*(\fg_X)$ in degree $-2$ (remembering the overall shift by $-2$);
\item[(2)] the term labeled by $\Bar{\epsilon} \d^2 z (y^\vee)$ contributes the summand $\clie^*(\fg_X ; \fg_X^\vee[-1])$ in degree $-1$;
\item[(3)] the term labeled by $ \epsilon \d^2 z (\Bar{y}^\vee)$ contributes the summand $\clie^*(\fg_X ; \fg_X^\vee[-1])$ in degree $-1$;
\item[(4)] the term labeled by $\d^2 z (y^\vee \Bar{y}^\vee)$ contributes the summand $\clie^*(\fg_X ; \fg_X^\vee [-1] \tensor \fg_X^\vee [-1] \oplus \fg_X^\vee[-1])$ in degree $0$;
\item[(5)] the term labeled by $\d^2 z (\d^2 y)^\vee$ contributes the summand $\clie^*(\fg_X ; \fg_X[1])$ in degree $-1$.
\end{itemize}

In all, we see that the $\CC^\times$ invariant subcomplex of (\ref{cplx2}) is of the form
\be\label{cplx2}
\xymatrix{
\ul{-2} & \ul{-1} & \ul{0} \\
\cred^*(\fg_X) \ar[r]^-{\d_1} \ar[dr]^-{\d_1} & \clie^*(\fg_X ; \fg_X^\vee[-1]) \ar[r]^-{\d_2} & \clie^*(\fg_X ; \fg_X^\vee [-1] \tensor \fg_X^\vee [-1]) \oplus \clie^*(\fg_X ; \fg_X^\vee[-1]) \\
& \clie^*(\fg_X ; \fg_X^\vee[-1]) \ar[ur]^-{\d_2} & \clie^*(\fg_X ; \fg_X[1])  \ar@{.>}[]!<1ex,-2ex>;[u]!<-6ex,1ex> \ar@{.>}[]!<1ex, 2ex>;[u]!<9ex,1ex>
}
\ee

We have seen that the differential of the deformation complex is of the form $\{S_h + S_X, -\} = \{S_h,-\} + \{S_X,-\}$.
Moreover, in the proof of Proposition \ref{CME} we showed that $\{S_h, -\}$ and $\{S_X, -\}$ commute and are each of square zero. 
The differential $\{S_X,-\}$ has the effect of turning on the Chevalley-Eilenberg differential for $\fg_X$ (and its modules). 
We now study what $\{S_h,-\}$ does to the complex above. 

Note that $S_h$ is a functional purely of the $\varphi$ variables. So, its Poisson bracket $\{S_h,-\}$ only acts non-trivially on the deformation complex containing functionals of the $\psi$ field. 
The only term above involving the $\psi$ field is term 5 which is given by $\clie^*(\fg_X ; \fg_X [1]) \cong \Omega^*_X(\sJ_{T_X})$.
The differential $\{S_h,-\}$ maps this term to the the factor $\clie^*(\fg_X ; \fg_X^\vee[-1]) \cong \Omega^*_X(\sJ_{T^\vee X})$ in term 4. 
In fact, this is nothing but the isomorphism $\Omega^*_X(\sJ_{T_X}) \cong \Omega^*_X(\sJ_{T^\vee X})$ determined by the metric $h$. 
Thus, the part of the differential $\{S_h,-\}$ has the effect of killing term 5 and the second factor of term 4 above.
We have included $\{S_h,-\}$ using the dotted arrow in the complex (\ref{cplx2}). 

Before stating the following lemma, we recall the following basic construction. 
The metric $h$ on $X$ induces the Levi-Civita connection on the tangent bundle $TX$ and also a covariant derivative on the cotangent bundle $T^*X$. 
At the level of sheaves, the covariant derivative is of the form
\ben
\nabla_h : \Omega^1_X \to \Omega^1_X \tensor \Omega^1_X .
\een 

\begin{lem}
There is a quasi-isomorphism of (\ref{cplx2}) (including the dotted arrows) and the complex
\be\label{cplx3}
\xymatrix{
\ul{-2} & \ul{-1} & \ul{0} \\
\sJ \sO_{red}(X) \ar[r]^-{\d_1} \ar[dr]^-{\d_1} & \sJ \Omega^1_X \ar[r]^-{\d_2'} & \sJ (\Omega^1_X \tensor \Omega^1_X)  \\
& \sJ \Omega^1_X \ar[ur]^-{\d_2'} & 
}
\ee
where $\d_2'$ is equal to the infinite jet expansion of the map
\ben
(\nabla_h, - \sigma \nabla_h) : \Omega^1_X \oplus \Omega^1_X \to \Omega^1_X \tensor \Omega^1_X
\een
sending $(\omega_1,\omega_2)$ to $\nabla_h \omega_1 - \sigma \nabla_h \omega_2$, with $\sigma$ the permutation operator sending $\omega_1 \tensor \omega_2 \mapsto \omega_2 \tensor \omega_1$. 
\end{lem}

Note that we can make the following change of coordinates in the degree $-1$ part of the complex (\ref{cplx3}) sending $(\omega_1,\omega_2) \mapsto (\omega_1+\omega_2,\omega_1-\omega_2)$ so that $\d_2'$ decouples into a symmetric and antisymmetric summand $\d_2' = {\rm Sym}^2 (\d_2') - \wedge^2 \d_2'$, where
\ben
\Sym^2 (\d_2') : \sJ \Omega^1_X \to \sJ \Sym^2 \Omega^1_X
\een
and
\ben
\wedge^2(\d_2') : \sJ \Omega^1_X \to \sJ \Omega^2_X .
\een 
We can check locally that former differential is nothing but the Lie derivative and the former the de Rham differential. (This is classical, see for instance Proposition 2.54 of \cite{GHL}.)

We are left with the following complex 
\ben
\xymatrix{
\sJ \sO_X^{red} \ar[r]^-{\d_{dR}} & \sJ \Omega^1_X \ar[r]^-{\d_{dR}} & \sJ \Omega^2_X \\
& \sJ T_X \ar[r]^-{L_{(-)} h} & \sJ \Sym^2 \Omega^1_X .
 }
 \een

\section{One-loop quantization}

In this section we construct the one-loop quantization of the classical BV theory introduced above. 
In the BV formalism this means that we construct a solution to the quantum master equation modulo $\hbar^2$. 
We proceed in the formalism developed in \cite{Cos1, cg1, cg2} which combines the effective functional approach to the path integral with the BV-BRST formalism for analyzing gauge symmetries. 
There are two steps to quantization.
First, we must write down an effective family of functionals satisfying the (homotopical) renormalization group equations. 
We utilize a regularization technique based on heat kernels to construct this effective family. 
Next, we must consider the {\em quantum master equation} (QME) for this effective family.
This can be viewed as studying quantizations that respect the natural gauge symmetries of the theory. 

We will only produce a quantization for the $\sigma$-model with source given by a two-dimensional disk. 

\subsection{Gauge fixing}
Before writing down the heat kernel and propagator we must choose a gauge fixing operator. 
For us, this is the degree $(-1)$ operator $Q^{GF}$ on the complex of fields $\sE$ given by
\ben
\xymatrix{
 \ul{-1} & \ul{0} \\
C^\infty_\Sigma \tensor \fg_X  & \ar[l]_-{\star \tensor h^\sharp} \Omega^2_\Sigma \tensor \fg_X^\vee 
}
\een 
where $\star$ is the Hodge star operator on $\Sigma$ and $h^\sharp$ is the inverse to $h^\flat$. 

The generalized Laplacian for this gauge fixing operator is
\begin{align*}
[Q, Q^{GF}] & = [\partial \Bar{\partial} \tensor h^\flat + 1 \tensor \ell_1, \star \tensor h^\sharp] \\
		   & = {\rm D} \tensor 1 + 1 \tensor [\ell_1, h^\sharp] .
\end{align*}
The operator ${\rm D}$ is the Laplace-Beltrami operator acting on functions and two-forms on $\Sigma$. 
Now, we have the freedom to choose the $L_\infty$ structure on $X$ to be compatible with the metric $h$. 
Indeed, we can choose the $L_\infty$ structure coming from the Levi-Civita connection $\nabla$ associated with $h$.
Once we do this, the operator $\ell_1$ is identified with $\nabla$.
Moreover, $[\ell_1, h^\sharp] = 0$ is equivalent to the condition $\nabla h = 0$. 
Thus, the generalized Laplacian is simply ${\rm D} \tensor 1$. 

\subsection{Heat kernel and the propagator}\label{sec:prop}
For $L > 0$ the heat kernel $K_L \in \sE \tensor \sE$ splits into three pieces. 
First, there is the scalar heat kernel for the Laplacian acting on functions on $\Sigma = \CC$ which has the explicit form
\ben
k_L(x,y) = \frac{1}{4 \pi L} e^{-|x-y|^2 / 4L} .
\een 
Second, there is the term $\d^2 x \tensor 1 - 1 \tensor \d^2 y$  which accounts for the natural pairing between functions and two-forms. 
Finally, there is the element $\id_{\fg_X} + \id_{\fg_X^\vee} \in \fg_X \tensor \fg_X^\vee \oplus \fg_X^\vee \tensor \fg_X$ accounting for the natural pairing between $\fg_X$ and its dual. 
In all, we write the heat kernel as
\ben
K_L(x,y) =  k_L(x,y) (\d^2 x \tensor 1 - 1 \tensor \d^2 y) (\id_{\fg_X} + \id_{\fg_X^\vee}) .
\een 
The defining property of the heat kernel is that it satisfies the relation
\ben
\<K_L(x,y), \Phi(y) \> = \left(e^{-L [Q,Q^{GF}]} \Phi\right) (x) 
\een
for any $\Phi \in \sE$. 

The propagator is defined by
\begin{align*}
P_{\epsilon \to L} (x,y) & = \int_{t = \epsilon}^L (Q^{GF} \tensor 1) K_t(x,y) \d t \\
& = \int_{t = \epsilon}^L \frac{1}{4 \pi t} e^{-|x-y|^2 / 4t} \d t \tensor h.
\end{align*}

As we discuss in Appendix \ref{app:app}, the propagator gives rise to a homotopy RG flow operator $W(P_{\epsilon \to L},-)$ and the corresponding homotopy Renormalization Group Equation (hRGE).
As we will see in Section \ref{sec:betaoverview}, the limit as $\epsilon \to 0$ of the hRG flow $W(P_{\epsilon \to L}, I)$ is not well-defined in general, and we must introduce counterterms as in \cite{Cos1}. In {\it loc. cit.}, it is shown that counterterms exist (given a choice of renormalization scheme), they are local functionals, and  they only depend on the UV parameter $\epsilon$.  Identification of the one-loop counterterms will play a critical role in the computation of the $\beta$-function for our theory; this is explained in the subsequent sections.

We can now define our na\"{i}ve one-loop quantization. 

\begin{dfn} 
The {\em na\"{i}ve} one-loop quantization of the two-dimensional $\sigma$-model of maps $\CC \to X$ is 
\ben
I[L] := \lim_{\epsilon \to 0} \sum_{\Gamma} W(P_{\epsilon \to L} , I + I^{CT}(\epsilon))
\een
where the sum is over all genus zero and one graphs $\Gamma$, and $I^{CT} (\epsilon)$ denotes a specific choice of counterterms dependent on fixing a renormalization scheme. 
\end{dfn}

This one-loop quantization indeed defines a {\it pre-theory} (modulo $\hbar^2$) in the parlance of \cite{Cos1}, i.e., the family of functionals $\{I[L]\}$ satisfies the hRGE (modulo $\hbar^2$) and is asymptotically local in the $L \to 0$ limit.
Further, $\{I[L]\}$ is a {\it pre-quantization} (to one loop) of the $\sigma$-model, as
\[
\lim_{L \to 0} I[L] \equiv I^h + I^X \; \text{  modulo } \; \hbar.
\]

\subsection{The quantum master equation}

To define a consistent QFT, the collection of interactions must satisfy the {\em quantum master equation}. 
This quantum master equation should be thought of as a quantization of the classical master equation we have already encountered.

The starting point is to introduce the operator $Q + \hbar \Delta_L$ where $\Delta_L : \sO(\sE) \to \sO(\sE)$ is the scale $L$ BV Laplacian defined by contraction with the scale $L$ heat kernel $K_L$. 
The issue here is that this operator is not square zero, in fact for every $L > 0$ one has
\ben
(Q + \hbar \Delta_L)^2 = \ell_1^2 .
\een

To rectify this, we introduce the modified $Q$-differential
\ben
Q_L = Q + \ell_1^2 \int_{t=0}^L Q^{GF} e^{-t {\rm D}} \d t .
\een

\begin{lem} 
For any $L>0$, the operator $Q_L$ satisfies
\ben
\left(Q_L + \hbar \Delta_L + \frac{1}{\hbar} F_{\ell_1} \right)^2 = C
\een
for some $C \in \Omega^\sharp_X$.
\end{lem}
\begin{proof}
This follows from the observations that $Q_L^2 = -\{F_{\ell_1},-\} = - \ell_1^2$ and $[\Delta_L, F_{\ell_1}] = \ell_1^2$.
\end{proof}

\begin{dfn} The family of functionals $\{I[L]\}$ is said to satisfy the scale $L$ quantum master equation (QME) if the operator
\ben
Q_L + \hbar \Delta_L + \{I[L],-\}_L 
\een 
is square zero acting on $\sO(\sE)$.
\end{dfn}

\begin{lem} The family of functionals $I[L]$ satisfies the scale $L$ QME if and only if 
\ben
Q_L I[L] + \hbar \Delta_L I[L] + \frac{1}{2} \{I[L],I[L]\}_L + F_{\ell_1} 
\een 
is zero modulo constant terms, i.e. elements in the ring $\Omega^\sharp_X$.
\end{lem}

\subsection{The obstruction}

Not every na\"{i}ve quantization of a classical field theory defines a quantization. 
The obstruction is precisely the failure to satisfy the quantum master equation we have just defined. 
In this section, we show that to order $\hbar$ the obstruction, though not identically zero, determines a trivial class in cohomology.
This is enough to determine the existence of a one-loop quantization of the classical $\sigma$-model.

\begin{prop} The obstruction to satisfying the quantum master equation modulo $\hbar^2$ is cohomologically trivial in the deformation complex $\Def$.
\end{prop} 

By Lemma C.5 in \cite{LiLi}, we see that the obstruction satisfies
\ben
\Theta [L] = \lim_{\epsilon \to 0} \sum_{\Gamma, e} W_{\Gamma,e}(P_{\epsilon}^L, K_\epsilon - K_0, I)^{sm} 
\een
where $``sm''$ denotes the smooth component of the functional. 
The sum is over one loop connected graphs $\Gamma$ and edges $e$ of the graph. 
The notation $W_{\Gamma, e}(A, B, I)$ where $A,B \in \Sym^2(\sE)$ means that we put $A$ on each edge besides the distinguished edge, where we put $B$. 

We find that the tadpole diagram is purely singular, hence does not contribute to $\Theta[L]$. 
Moreover, the diagrams involving three or more vertices vanish by type reasons. 
That leaves the wheel with two vertices with internal edges labeled by $P_{\epsilon \to L}$ and $K_{\epsilon} - K_0$, 
as shown here:
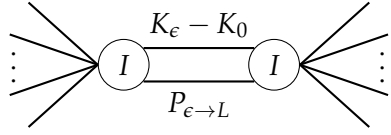
\begin{center}
        \begin{tikzpicture}
\node[circle] (l) at (2,0) [draw] {$I$};
\node[circle] (r) at (4,0) [draw] {$I$};

 \node (uul) at (0.6,0.95) {};
 \node (ul) at (0.35, 0.45) {};
  \node (dl) at (0.35,-0.45) {};
  \node (ddl) at (0.6,-0.95) {};
\node[rotate=90] at (0.55,0) {$\dotsb$};

 \node (uur) at (5.4,0.95) {};
 \node (ur) at (5.65, 0.45) {};
  \node (dr) at (5.65,-0.45) {};
  \node (ddr) at (5.4,-0.95) {};
\node[rotate=90] at (5.45,0) {$\dotsb$};

\draw[thick] (l.40) --  (r.140) node[midway,above] {$K_{\epsilon} - K_0$};
\draw[thick] (l.-40) -- (r.-140) node[midway,below] {$P_{\epsilon \to L}$};

\draw[thick] (uul)--(l.180);
\draw[thick] (ul)--(l.180);
\draw[thick] (dl) --(l.180);
\draw[thick] (ddl) --(l.180);

\draw[thick] (uur)--(r.0);
\draw[thick] (ur)--(r.0);
\draw[thick] (dr) --(r.0);
\draw[thick] (ddr) --(r.0);

        \end{tikzpicture}
\end{center}

The non-vanishing part of the weight of the diagram is of the form
\ben
\int_{z,w} \int_{t = \epsilon}^L \frac{1}{4 \pi \epsilon} \frac{1}{4 \pi t} e^{-|z-w|^2 / 4t} \varphi(z,w) \d^2 z \d^2 w \d t
\een
where $\varphi$ is some compactly supported function on $\CC^2$ depending on the inputs of the diagram. 
Applying Wick's formula, we see that as $L \to 0$ the smooth part of the diagram can be written as the following local functional
\ben
- \frac{\log 2}{4\pi} \Theta (\varphi,\psi)
\een
where $\Theta$ is represented by the diagram:

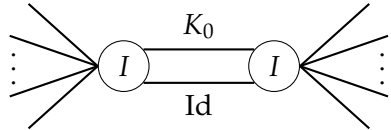
\begin{center}
        \begin{tikzpicture}
\node[circle] (l) at (2,0) [draw] {$I$};
\node[circle] (r) at (4,0) [draw] {$I$};

 \node (uul) at (0.6,0.95) {};
 \node (ul) at (0.35, 0.45) {};
  \node (dl) at (0.35,-0.45) {};
  \node (ddl) at (0.6,-0.95) {};
\node[rotate=90] at (0.55,0) {$\dotsb$};

 \node (uur) at (5.4,0.95) {};
 \node (ur) at (5.65, 0.45) {};
  \node (dr) at (5.65,-0.45) {};
  \node (ddr) at (5.4,-0.95) {};
\node[rotate=90] at (5.45,0) {$\dotsb$};

\draw[thick] (l.40) --  (r.140) node[midway,above] {$K_0$};
\draw[thick] (l.-40) -- (r.-140) node[midway,below] {${\rm Id}$};

\draw[thick] (uul)--(l.180);
\draw[thick] (ul)--(l.180);
\draw[thick] (dl) --(l.180);
\draw[thick] (ddl) --(l.180);

\draw[thick] (uur)--(r.0);
\draw[thick] (ur)--(r.0);
\draw[thick] (dr) --(r.0);
\draw[thick] (ddr) --(r.0);

        \end{tikzpicture}
\end{center}

It suffices to show that we can make this functional exact in the deformation complex. 
For this, consider the following local functional 

\begin{center}
        \begin{tikzpicture}
        
\node at (5,0) {$J (\varphi,\psi) \;\;\;= $};        
\node[circle] (c) at (8,0) [draw] {$I$};

 \node (uulc) at (6.6,0.95) {};
 \node (ulc) at (6.35, 0.45) {};
  \node (llc) at (6.35,-0.45) {};
  \node (lllc) at (6.6,-0.95) {};
\node[rotate=90] at (6.55,0) {$\dotsb$};
\tikzset{every loop/.style={min distance=10mm,in=-35,out=35,looseness=7}};
\path[thick] (c) edge  [loop above]  ();
\node at (9.25,0) {${\rm Id}$};

\draw[thick] (ulc)--(c.180);
\draw[thick] (llc)--(c.180);
\draw[thick] (uulc) --(c.180);
\draw[thick] (lllc) --(c.180);

        \end{tikzpicture}
\end{center}

\noindent
The differential of the deformation complex is of the form $\{S,-\} = Q + \{I,-\}$. 
Applying this operator to the functional $J$ we find

\begin{center}
 \begin{tikzpicture}
 
 \node at (1.5,-4) {$(6)$};
        
\node at (4.5,-4) {$QJ + \{I,J\} = $};       
\node at (6,-4) {$Q$};
\node[circle] (c) at (8,-4) [draw] {$I$};

 \node (uulc) at (6.6,-3.05) {};
 \node (ulc) at (6.35, -3.55) {};
  \node (llc) at (6.35,-4.45) {};
  \node (lllc) at (6.6,-4.95) {};
\node[rotate=90] at (6.55,-4) {$\dotsb$};
\tikzset{every loop/.style={min distance=10mm,in=-35,out=35,looseness=7}};
\path[thick] (c) edge  [loop above]  ();
\node at (9.25,-4) {${\rm Id}$};

\node (plu) at (6.45,-3) {};
\node (pll) at (6.45,-5){};
\draw[thick] (plu) to[bend right=15] (pll);

\node (pru) at (9.45,-3) {};
\node (prl) at (9.45,-5){};
\draw[thick] (pru) to[bend left=15] (prl);

\draw[thick] (ulc)--(c.180);
\draw[thick] (llc)--(c.180);
\draw[thick] (uulc) --(c.180);
\draw[thick] (lllc) --(c.180);

  \node at (10,-4) {$+$};

\end{tikzpicture}
 \begin{tikzpicture}

\node[circle] (l) at (2,0) [draw] {$I$};
\node[circle] (r) at (4,0) [draw] {$I$};

 \node (uul) at (0.6,0.95) {};
 \node (ul) at (0.35, 0.45) {};
  \node (dl) at (0.35,-0.45) {};
  \node (ddl) at (0.6,-0.95) {};
\node[rotate=90] at (0.55,0) {$\dotsb$};

 \node (uur) at (5.4,0.95) {};
 \node (ur) at (5.65, 0.45) {};
  \node (dr) at (5.65,-0.45) {};
\node[rotate=90] at (5.45,0) {$\dotsb$};

\draw[thick] (l.0) --  (r.180) node[midway,above] {$K_0$};

\draw[thick] (uul)--(l.180);
\draw[thick] (ul)--(l.180);
\draw[thick] (dl) --(l.180);
\draw[thick] (ddl) --(l.180);

\draw[thick] (uur)--(r.0);
\draw[thick] (ur)--(r.0);
\draw[thick] (dr) --(r.0);

\tikzset{every loop/.style={min distance=10mm,in=240,out=300,looseness=7}};
\path[thick] (r) edge  [loop above]  ();

\node at (4, -1.1) {${\rm Id}$};

\end{tikzpicture}

\end{center}

\noindent
Note that the functional in the parentheses is equal to $\Tilde{\Delta} (I)$ where $\Tilde{\Delta}$ is the BV Laplacian. 
Moreover, one has 

\begin{center}
 \begin{tikzpicture}
\node at (-0.25,0) {$\Tilde{\Delta}$};
\node[circle] (l) at (2,0) [draw] {$I$};
\node[circle] (r) at (4,0) [draw] {$I$};

 \node (uul) at (0.6,0.95) {};
 \node (ul) at (0.35, 0.45) {};
  \node (dl) at (0.35,-0.45) {};
  \node (ddl) at (0.6,-0.95) {};
\node[rotate=90] at (0.55,0) {$\dotsb$};

\node (plu) at (0.35,1.1) {};
\node (pll) at (0.35,-1.1){};
\draw[thick] (plu) to[bend right=15] (pll);

 \node (uur) at (5.4,0.95) {};
 \node (ur) at (5.65, 0.45) {};
  \node (dr) at (5.65,-0.45) {};
  \node (ddr) at (5.4,-0.95) {};
\node[rotate=90] at (5.45,0) {$\dotsb$};

\node (pru) at (5.65,1.1) {};
\node (prl) at (5.65,-1.1){};
\draw[thick] (pru) to[bend left=15] (prl);

\draw[thick] (l.0) --  (r.180) node[midway,above] {$K_0$};

\draw[thick] (uul)--(l.180);
\draw[thick] (ul)--(l.180);
\draw[thick] (dl) --(l.180);
\draw[thick] (ddl) --(l.180);

\draw[thick] (uur)--(r.0);
\draw[thick] (ur)--(r.0);
\draw[thick] (dr) --(r.0);
\draw[thick] (ddr) --(r.0);

\node at (6.2,0) {$=$};

\node[circle] (l) at (8,1.7) [draw] {$I$};
\node[circle] (r) at (10,1.7) [draw] {$I$};

 \node (uul) at (6.6,2.65) {};
 \node (ul) at (6.35, 2.15) {};
  \node (dl) at (6.35,1.25) {};
  \node (ddl) at (6.6,.75) {};
\node[rotate=90] at (6.55,1.7) {$\dotsb$};

 \node (uur) at (11.4,2.65) {};
 \node (ur) at (11.65, 2.15) {};
  \node (dr) at (11.65,1.25) {};
\node[rotate=90] at (11.45,1.7) {$\dotsb$};

\draw[thick] (l.0) --  (r.180) node[midway,above] {$K_0$};

\draw[thick] (uul)--(l.180);
\draw[thick] (ul)--(l.180);
\draw[thick] (dl) --(l.180);
\draw[thick] (ddl) --(l.180);

\draw[thick] (uur)--(r.0);
\draw[thick] (ur)--(r.0);
\draw[thick] (dr) --(r.0);

\tikzset{every loop/.style={min distance=10mm,in=240,out=300,looseness=7}};
\path[thick] (r) edge  [loop above]  ();

\node at (10, 0.6) {${\rm Id}$};

\node at (7,0) {$+$};

\node[circle] (l) at (9,-0.5) [draw] {$I$};
\node[circle] (r) at (11,-0.5) [draw] {$I$};

 \node (uul) at (7.6,0.45) {};
 \node (ul) at (7.35, -0.05) {};
  \node (dl) at (7.35,-0.95) {};
  \node (ddl) at (7.6,-1.45) {};
\node[rotate=90] at (7.55,-0.5) {$\dotsb$};

 \node (uur) at (12.4,0.45) {};
 \node (ur) at (12.65, -0.05) {};
  \node (dr) at (12.65,-0.95) {};
  \node (ddr) at (12.4,-1.45) {};
\node[rotate=90] at (12.45,-0.5) {$\dotsb$};

\draw[thick] (l.40) --  (r.140) node[midway,above] {$K_0$};
\draw[thick] (l.-40) -- (r.-140) node[midway,below] {${\rm Id}$};

\draw[thick] (uul)--(l.180);
\draw[thick] (ul)--(l.180);
\draw[thick] (dl) --(l.180);
\draw[thick] (ddl) --(l.180);

\draw[thick] (uur)--(r.0);
\draw[thick] (ur)--(r.0);
\draw[thick] (dr) --(r.0);
\draw[thick] (ddr) --(r.0);

\end{tikzpicture} 
\end{center}

\noindent
Since $Q$ and $\Tilde{\Delta}$ commute, we find that the right-hand side of Equation (6) becomes

\begin{center}
\begin{tikzpicture}

\node at (-1.25,0) {$\Tilde{\Delta}$};

\node at (-0.25,0) {$QI +$};

\node[circle] (l) at (2,0) [draw] {$I$};
\node[circle] (r) at (4,0) [draw] {$I$};

 \node (uul) at (0.6,0.95) {};
 \node (ul) at (0.35, 0.45) {};
  \node (dl) at (0.35,-0.45) {};
  \node (ddl) at (0.6,-0.95) {};
\node[rotate=90] at (0.55,0) {$\dotsb$};

\node (plu) at (-0.65,1.1) {};
\node (pll) at (-.65,-1.1){};
\draw[thick] (plu) to[bend right=15] (pll);

 \node (uur) at (5.4,0.95) {};
 \node (ur) at (5.65, 0.45) {};
  \node (dr) at (5.65,-0.45) {};
  \node (ddr) at (5.4,-0.95) {};
\node[rotate=90] at (5.45,0) {$\dotsb$};

\node (pru) at (5.65,1.1) {};
\node (prl) at (5.65,-1.1){};
\draw[thick] (pru) to[bend left=15] (prl);

\draw[thick] (l.0) --  (r.180) node[midway,above] {$K_0$};

\draw[thick] (uul)--(l.180);
\draw[thick] (ul)--(l.180);
\draw[thick] (dl) --(l.180);
\draw[thick] (ddl) --(l.180);

\draw[thick] (uur)--(r.0);
\draw[thick] (ur)--(r.0);
\draw[thick] (dr) --(r.0);
\draw[thick] (ddr) --(r.0);

\node at (6.3,0) {$-$};

\node[circle] (l) at (8.2,0) [draw] {$I$};
\node[circle] (r) at (10.2,0) [draw] {$I$};

 \node (uul) at (6.8,0.95) {};
 \node (ul) at (6.55, 0.45) {};
  \node (dl) at (6.55,-0.45) {};
  \node (ddl) at (6.8,-0.95) {};
\node[rotate=90] at (6.75,0) {$\dotsb$};

 \node (uur) at (11.6,0.95) {};
 \node (ur) at (11.85, 0.45) {};
  \node (dr) at (11.85,-0.45) {};
  \node (ddr) at (11.6,-0.95) {};
\node[rotate=90] at (11.65,0) {$\dotsb$};

\draw[thick] (l.40) --  (r.140) node[midway,above] {$K_0$};
\draw[thick] (l.-40) -- (r.-140) node[midway,below] {${\rm Id}$};

\draw[thick] (uul)--(l.180);
\draw[thick] (ul)--(l.180);
\draw[thick] (dl) --(l.180);
\draw[thick] (ddl) --(l.180);

\draw[thick] (uur)--(r.0);
\draw[thick] (ur)--(r.0);
\draw[thick] (dr) --(r.0);
\draw[thick] (ddr) --(r.0);

\end{tikzpicture} 
\end{center}

\noindent
Finally, we note that the expression inside the parentheses is zero by the classical master equation. 
We conclude that
\ben
(Q + \{I,-\}) \left(\frac{\log 2}{2 \pi} J\right) = \Theta,
\een
as desired. 
We note that $J$ has scaling weight zero. 

To obtain the quantization from the na\"{i}ve quantization we consider the family of functionals 
\[
\{I[L] + \hbar J [L]\}.
\]
Here, $I[L]$ is, as above, equal to a sum over graphs of genus $\leq 1$. 
Since we are working modulo $\hbar^2$, the quantity $\hbar J[L]$ is equal to the sum of the weights of trees where at a single vertex we put the functional $\hbar J$ and at the remaining vertices we put the functional $I$. 
It is an immediate consequence of the compatibility of homotopy RG-flow and the quantum differential that this family of functionals satisfies the quantum master equation modulo $\hbar^2$.

\section{$\beta$-function generalities}\label{sec:betaoverview}

Having defined and quantized (to one-loop) the $\sigma$-model, we now recall some generalities and computational lemmas regarding $\beta$-functions. We compute the $\beta$-function for our quantization in Section \ref{sect:betasigma} below.

Consider a translation invariant BV theory $(E,Q, I)$ on $\RR^n$.
We denote by $E_0$ the fiber of the bundle $E$ over $0 \in \RR^n$. 
The group $\RR_{>0}$ acts on the fields $\sE = C^\infty(\RR^n) \tensor E_0$ via the action induced by rescaling $\RR^n$. 
Further, $\RR_{>0}$ acts on the space of functionals, $\sO(\sE)$, and preserves the subspace of local functionals $\sO_{loc} (\sE)$. We will denote this action of $\RR_{>0}$ on $\sO(\sE)$ by $\rho_\lambda$. 
Making this action explicit, one can easily check the following :
\begin{equation}\label{eqn:weight}
\rho_\lambda (I_k) = \lambda^{n+\frac{k(2-n)}{2}}I_k, \quad \text{ where } \quad I_k (\phi) = \int_{\RR^n} \phi (x)^k.
\end{equation}

\begin{dfn}
A classical field theory $(E,Q,I)$ is {\em scale-invariant} if $\rho_\lambda (I) = I$ for all $\lambda \in \RR_{>0}$.
\end{dfn}

From equation (\ref{eqn:weight}) above, one can see that $\phi^4$ theory is scale-invariant on $\RR^4$ and $\phi^3$ theory is scale-invariant on $\RR^6$. Scale-invariance is a stronger condition than simply being {\it (classically) renormalizable}, which is the condition that $\rho_\lambda (I)$ flows to a fixed point as $\lambda \to 0$.

We now recall how to extend the action of $\RR_{>0}$ to QFTs in the BV formalism.
This action is called the {\it (local) renormalization group flow}, or simply RG flow. 
We start with a classical BV theory given by the data $(E, Q, I)$.
Following Costello \cite{Cos1}, a quantum field theory is a collection of functionals $\{I[L]\}$ where $L > 0$ that satisfy homotopy RG flow (the renormalization group equation) and satisfy the quantum master equation.  
It is a quantization if modulo $\hbar$ the family of functionals converges to the classical interaction as $L \to 0$. 
It is also necessary to equip our quantum BV theory with a gauge-fixing operator $Q^{GF}$. 
Further, suppose that the action of $\RR_{>0}$ on $Q^{GF}$ is as follows:
\[
\rho_\lambda \cdot Q^{GF} := \rho_\lambda Q^{GF} \rho_{\lambda^{-1}} = \lambda^k Q^{GF}.
\]
Here the {\it scaling dimension} of $Q^{GF}$ is allowed to be a rational number, i.e., $k \in \QQ$. In the scalar theories below, $k = n/2-1$.

\begin{dfn}
Let $\{I[L]\}$ be a family of translation-invariant effective interactions.
Define a rescaled effective family $\{I_\lambda [L]\}$ by
\[
I_\lambda [L] := \rho_\lambda \cdot I[\lambda^{-k} L].
\]
\end{dfn}

It can be shown explicitly, Lemma 3.4 of \cite{EWY}, that $\{I[L]\}$ satisfies homotopy RG flow and the QME if and only if the rescaled family $\{I_\lambda [L]\}$ does.

We are (finally) in a position to define the $\beta$-functional.  We will distinguish between the $\beta$-functional and $\beta$-function as the latter is the cohomology class of the former in the BV algebra of (quantum) observables.

\begin{dfn}
Let $\{I[L]\}$ be an effective quantization of the BV theory $(E,Q,I)$ on $\RR^n$. For $L \in \RR_{>0}$, define the {\it scale $L$ $\beta$-functional} to be the functional
\[
\sO_\beta[L] := \lim_{\lambda \to 1} \lambda \frac{d}{d \lambda} I_\lambda [L].
\]
\end{dfn}

We can expand the $\beta$-functional in powers of $\hbar$:
\[
\sO_\beta [L] = \sO^{(0)}_\beta [L] + \hbar \sO^{(1)}_\beta [L] + O(\hbar^2).
\]
The superscript indicates the loop depth, e.g., $\sO^{(1)}_\beta [L]$ is the {\it one-loop $\beta$-functional}. For a theory where the classical theory is scale invariant, the zero-loop $\beta$-functional, $\sO^{(0)}_\beta [L]$, is identically zero.

\begin{prop}[Corollary 3.12 of \cite{EWY}]\label{prop:oneloop}
Let $\{I[L]\}$ be an effective quantization of the translation and scale invariant BV theory $(E,Q,I)$ on $\RR^n$.  Then
\[
\sO^{(1)}_\beta := \lim_{L \to 0} \sO^{(1)}_\beta [L]
\]
exists and determines a closed element in $\sO_{loc} (\sE)$.
\end{prop}

\begin{rmk}\label{rmk:warning}
The higher loop $\beta$-functionals, $\sO^{(k)}_\beta = \lim_{L \to 0} \sO^{(k)}_\beta [L]$, are not necessarily well-defined. 
Even if one na\"{i}vely defines $\sO^{(k)}_\beta$ to be the $\log \epsilon$ divergence at $k$-loops (compare Proposition \ref{prop:logdiv}), these functionals are not closed with respect to the BRST differential $Q + \{I,-\}$. 
One can prove that if $\sO^{(i)}_\beta \equiv 0$ for $i < k$, then $\sO^{(k)}_\beta [L]$ satisfies homotopy RG flow and is BRST closed; in particular, the $k$-loop $\beta$-functional, $\sO^{(k)}_\beta$, exists.
\end{rmk}

\begin{dfn}
Let $\{I[L]\}$ be an effective quantization of a translation invariant BV theory $(E,Q,I)$ on $\RR^n$. The {\it scale $L$ $\beta$-function}, $\beta[L]$, is the cohomology class
\[
\beta[L] := [ \sO_\beta [L]]\in H^0 (\sO(\sE) \otimes C^\infty((\epsilon, L))).
\]
\end{dfn}

\begin{rmk} Typically, the $\beta$-function cannot be decomposed by $\hbar$ degrees; the complex of functionals is only filtered, not graded.
\end{rmk}

In light of the previous remark, we do have the following (a corollary of Proposition \ref{prop:oneloop}).

\begin{prop}
Let $\{I[L]\}$ be an effective quantization of the translation and scale invariant BV theory $(E,Q,I)$ on $\RR^n$.  Then
\[
\beta^{(1)} := \left [\sO^{(1)}_\beta \right ] \in H^0 (\sO_{loc}(\sE))
\]
is well-defined.
\end{prop}

\subsection{As a local functional}\label{sect:betafunc}

Consider a BV theory $(E,Q,I)$ and the associated BRST complex $(\sO_{loc} (\sE), Q + \{I,-\})$. Cohomology classes in this complex define first order deformations of the theory. Hence, we could think of $H^0 (\sO_{loc} (\sE))$ as the space of {\it coupling constants} for the given theory. Now if $(E,Q,I)$ is translation and scale invariant, then the one-loop $\beta$-functional is BRST closed. Consequently, we can think of the one-loop $\beta$-function as a function on $H^0 (\sO_{loc} (\sE))$. 

In order to derive an explicit formula, we should fix {\it bare values} (a basis) for the space $H^0 (\sO_{loc} (\sE))$, and we should also check how $\beta^{(1)}$ transforms under a change of basis (Proposition 3.18 of \cite{EWY}). Moreover, we typically would like to distinguish a finite dimensional subspace of $H^0 (\sO_{loc}(\sE))$ and realize $\beta^{(1)}$ as a function of a $\RR^n$ valued parameter ${\bf c}$. In some good cases, e.g., $\phi^3$ theory below, we can find a one-dimensional subspace and $\beta^{(1)}$ is a function of a real number $c$.

Ideally, we would like to choose the distinguished  subspace to be the one spanned by the classical interaction $I$ under homotopy $RG$ flow.  However, as we discuss in Section \ref{sect:dynamic}, homotopy RG flow may not actually preserve this subspace and can introduce {\it dynamic coupling constants}. We can however restrict to the subspace spanned by functionals which are $\log \epsilon$ divergent in the $\epsilon \to 0$ limit to obtain a well-defined function of just a few coupling constants.

\subsection{Computational lemmas}

The gloss above illustrates the development and interpretation of the $\beta$-function/functional. However, if one is interested in computations at one-loop, they can utilize the following result.

\begin{prop}[Proposition 3.23 of \cite{EWY}]\label{prop:logdiv}
Let $(E,Q,I)$ be a translation and scale-invariant BV theory and $\{I[L]\}$ an effective quantization at one-loop. Let $k$ denote the scaling dimension of $Q^{GF}$. Then the one-loop $\beta$-functional satisfies
\[
\sO^{(1)}_\beta = k I^{CT}_{\log}.
\]
\end{prop}

Finally, we note a convenient property of the $\beta$-function: homotopy invariance.

\begin{prop}[Proposition 3.20 of \cite{EWY}]
The $\beta$-function is locally constant on the space of quantum field theories with fixed classical BV complex and gauge fixing operator.
\end{prop}

As an example, the first and second order formulations of Yang-Mills theory are homotopy equivalent, so the $\beta$-function can be computed using the first order formalism.

\subsection{First examples: scalar field theory}\label{sect:scalar}

Building on \cite{Cos1}, we describe the $\beta$-functional/function for certain scalar field theories on $\RR^n$. That is, our theory has field content
\[
\sE : C^\infty_c (\RR^n) \xrightarrow{\; D \;} C^\infty_c (\RR^n)[-1],
\]
where $D$ is the Laplacian, and action 
\[
S(\phi) = \int_{\RR^n} \phi D \phi + I (\phi), \text{ for } \phi \in \sE.
\]
In this setting the propagator is determined by the heat kernel $K_\ell (x_1, x_2)$ on $\RR^n$.

We will choose a renormalization scheme where $\epsilon^{-1}$ and $\log \epsilon$ are purely singular. Moreover, we will consider the space $\cS \subset H^{0} (\sO_{loc} (\sE))$ spanned by those functionals which are $\log \epsilon$ divergent in the $\epsilon \to 0$ limit.  In the theories below, one can verify that $\cS$ is one-dimensional and we will choose the classical interaction functional $I$ as a basis vector, i.e., $\cS = \{c I : c \in \RR\}$.

\subsubsection{$\phi^4$ theory on $\RR^4$}

Consider the scalar theory on $\RR^4$ with interaction
\[
I(\phi)= I_{0,4}  = \frac{1}{4!} \int_{x \in \RR^4} \phi(x)^4 .
\]
The only term that contributes to one-loop $\log \epsilon$ divergence, and hence the one-loop $\beta$-functional via Proposition \ref{prop:logdiv}, comes from the following graph.

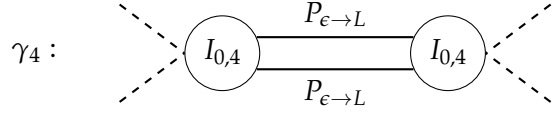
\vspace{1ex}

\begin{center}
        \begin{tikzpicture}
 \node (ul) at (-1.5,0.75) {};
  \node (ll) at (-1.5,-0.75) {};
 \node (ur) at (4.5,0.75) {};
 \node (lr) at (4.5,-0.75) {};
\node at (-2.5,0) {$\gamma_4 :$};

\node[circle] (l) at (0,0) [draw] {$I_{0,4}$};
\node[circle] (r) at (3,0) [draw] {$I_{0,4}$};
\draw[thick] (l.25) --  (r.155) node[midway,above] {$P_{\epsilon \to L}$};
\draw[thick] (l.-25) --  (r.-155) node[midway,below] {$P_{\epsilon \to L}$};
\draw[dashed,thick] (ul) -- (l.180);
\draw[dashed,thick] (ll) -- (l.180);
\draw[dashed,thick] (ur)--(r.0);
\draw[dashed,thick] (lr)--(r.0);

        \end{tikzpicture}
\end{center}

The corresponding counterterm is computed by Costello in Section 4, Chapter 4 of \cite{Cos1}. Therefore we have,
\[
\sO_\beta^{(1)} = I^{CT}_{1,4} (\epsilon) = - \pi^{-2} 2^{-8} c^2 \log \epsilon \int_{x \in \RR^4} \phi(x)^4.
\]

\subsubsection{$\phi^3$ theory on $\RR^6$}

Again, in the case 
\[
I(\phi) = I_{0,3} = \frac{1}{3!} \int_{x \in \RR^6} \phi(x)^3,
\]
there is only one term that contributes to the one-loop $\log \epsilon$ divergence: the wheel with three vertices.

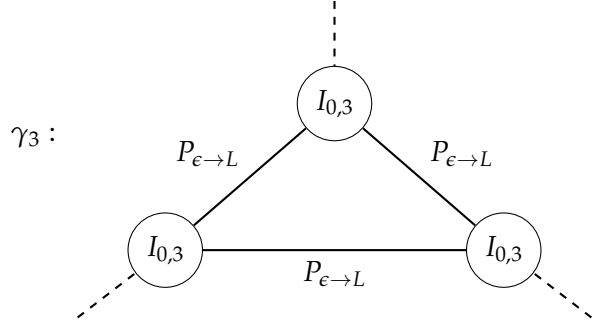
\vspace{1ex}

\begin{center}
        \begin{tikzpicture}
 \node (u) at (2.25, 3.45) {};
  \node (ll) at (-1.3,-1) {};
 \node (lr) at (5.8,-1) {};
\node at (-1.75,1.5) {$\gamma_3 :$};

\node[circle] (l) at (0,0) [draw] {$I_{0,3}$};
\node[circle] (r) at (4.5,0) [draw] {$I_{0,3}$};
\node[circle] (t) at (2.25,1.95) [draw] {$I_{0,3}$};

\draw[thick] (l) --  (r) node[midway,below] {$P_{\epsilon \to L}$};
\draw[thick] (l) --  (t) node[midway,above left] {$P_{\epsilon \to L}$};
\draw[thick] (t) -- (r) node[midway, above right] {$P_{\epsilon \to L}$};
\draw[dashed,thick] (u) -- (t);
\draw[dashed,thick] (ll) -- (l);
\draw[dashed,thick] (lr)--(r);
        \end{tikzpicture}
\end{center}

Let us compute the weight of this graph with respect to the propagator $P_{\epsilon \to L}$.

\begin{align*}
 \MoveEqLeft[2] \omega_{\gamma_3} \\[1ex]
 =& \int_{\ell_1 , \ell_2 , \ell_3 \in [\epsilon, L]}  \int_{x_1, x_2, x_3 \in \RR^6} K_{\ell_1} (x_1, x_2) K_{\ell_2} (x_2, x_3) K_{\ell_3} (x_3, x_1)  d{\bf x} d{\bf \ell}\\[1ex]
=& \frac{1}{(4 \pi)^9} \int_{\ell_1 , \ell_2 , \ell_3 \in [\epsilon, L]}  \int_{x_1, x_2, x_3 \in \RR^6} (\ell_1 \ell_2 \ell_3 )^{-3} e^{- \lVert x_1 - x_2 \rVert^2/4 \ell_1 - \lVert x_2-x_3 \rVert^2/4 \ell_2 - \lVert x_3-x_1 \rVert^2/4 \ell_3}   d{\bf x} d{\bf \ell}\\[1ex]
\overset{(1)}{=}&\frac{3}{8(4 \pi)^9} \int_{\ell_1 , \ell_2 , \ell_3 \in [\epsilon, L]}  \int_{y, z_1, z_2 \in \RR^6} (\ell_1 \ell_2 \ell_3)^{-3}e^{-\lVert z_1 \rVert^2/\ell_1 - \lVert z_2 \rVert^2/\ell_2 - \lVert z_1 +z_2 \rVert^2/\ell_3}  dy d{\bf z} d {\bf \ell}\\[1ex]
\overset{(2)}{=}& \frac{3}{2^{21} \pi^3} \int_{\ell_1 , \ell_2 , \ell_3 \in [\epsilon, L]} (\ell_1 \ell_2 \ell_3)^{-3} \left ( \ell_1^{-1} \ell_2^{-1} + \ell_1^{-1} \ell_3^{-1} + \ell_2^{-1} \ell_3^{-1} \right )^{-3} d {\bf \ell}\\[1ex]
=& \frac{3}{2^{21} \pi^3} \int_{\ell_1 , \ell_2 , \ell_3 \in [\epsilon, L]} \frac{1}{(\ell_1 + \ell_2 +\ell_3)^3} d {\bf \ell}.
\end{align*}

\vspace{1ex}

\noindent
In the equality $(1)$ we have used the coordinate substitution
\[
z_1 = \frac{1}{2} (x_2-x_3) , \quad z_2 = \frac{1}{2} (x_3 -x_1), \quad y = \frac{1}{2} (x_1+x_2+x_3).
\]
Further, in equality $(2)$ we have used the equality
\[
\int_{y, z_1, z_2 \in \RR^6} (\ell_1 \ell_2 \ell_3)^{-3}e^{-\lVert z_1 \rVert^2/\ell_1 - \lVert z_2 \rVert^2/\ell_2 - \lVert z_1 +z_2 \rVert^2/\ell_3}  dy d{\bf z} = \pi^6 (\det A)^{-3},
\]
where $A$ is the matrix corresponding to the quadratic form
\[
\lVert z_1 \rVert^2/\ell_1 + \lVert z_2 \rVert^2/\ell_2 + \lVert z_1 +z_2 \rVert^2/\ell_3.
\]

Actually, all we care about is the $\log \epsilon$ divergence of the graph weight $\omega_{\gamma_3}$:
\[
\mathrm{sing}_{\log \epsilon} \omega_{\gamma_3} = \frac{3}{2^{22} \pi^3} \log \epsilon .
\]
Therefore, after accounting for graph automorphisms, we find that the $\log \epsilon$ counterterm and hence the one loop $\beta$ functional is given by
\[
\sO_\beta^{(1)} = 2 I^{CT}_{1,3}(\epsilon) =  2^{-21} \pi^{-3} c^3 \log \epsilon \int_{x \in \RR^6} \phi (x)^3.
\]

\subsubsection{Turning on a mass} 

When one introduces a mass term, $m^2 \int_{\RR^n} \phi (x)^2$, the theory is often no longer scale invariant. One could choose to ignore this subtlety and formally compute the one loop $\beta$-functional as the $\log \epsilon$ divergence at one loop.

In a massive theory, the heat kernel (and hence propagator) changes as follows
\[
K_\ell^m (x_1 , x_2) = K_\ell^{m=0} (x_1, x_2) e^{-tm^2},
\]
where $K_\ell^m$ is the kernel for the massive theory and $K_\ell^{m=0}$ is the kernel for the massless theory.

One can then show directly that for any admissible graph $\gamma$ and interaction $I$, we have
\[
\mathrm{sing}_{\{\epsilon^{-1} , \log \epsilon\}} \omega (P^m_{\epsilon \to L}, I) (\phi) = \mathrm{sing}_{\log \epsilon} \omega (P^{m=0}_{\epsilon \to L}, I) (\phi) + O (\epsilon^{-1}).
\]
Consequently, a mass introduces no further $\log \epsilon$ corrections and hence doesn't change the one loop $\beta$-functional/function.

\subsubsection{Dynamic coupling constants}\label{sect:dynamic}

We note that renormalization/homotopy renormalization group flow can introduce {\it dynamic coupling constants}. That is, homotopy RG flow does not preserve the $\RR$-linear subspace of $\sO_{loc} (\sE)$ spanned by the original (classical) interaction $I^{CL} = \sum_n I_{0,n}$.  This phenomenom already appears in $\phi^3$ theory on $\RR^6$ (and $\phi^4$ theory on $\RR^4$).

Indeed, in $\phi^3$ theory on $\RR^6$, we find that
\[
I_{1,2}^{CT}[\epsilon] \sim \epsilon^{-1} \phi (x)^2.
\]
This term could be interpreted as a {\it dynamic mass}. As noted above, these mass terms are not scale invariant; in $\phi^3$ theory this term has scaling weight 2.

There are (at least) two ways to account for this defect/feature of  renormalization in terms of the $\beta$-functional.  First, one could impose additional symmetry so that these dynamic terms are not invariant.  Secondly, one could note that flow generated by the $\beta$-functional  (suitably interpreted as a vector field on $\sO_{loc} (\sE)$) preserves the subspace of local functionals which are $\log \epsilon$ divergent in the $\epsilon \to 0$ limit. We will take the second approach in what follows; the former way is illustrated in the case of Yang-Mills in \cite{EWY}.

\section{The one-loop $\beta$-function for the $\sigma$-model}\label{sect:betasigma}

We now use the framework from the preceding section to the compute the one-loop $\beta$ function for our theory.
Recall that the classical action consists of two pieces, $S = S_h + S_X$. Therefore, there are two flavors of classical interaction terms (vertex types): $I^h_{0,k}$ and $I^X_{0,\ell}$.  As noted in Section \ref{sec:prop}, the propagator only pairs the $\phi$ fields, and one can check that the (singular) BV operator $\Delta_L$ pairs $\phi$ and $\psi$ fields.
Diagramatically, we can represent these two functionals/operators as follows.

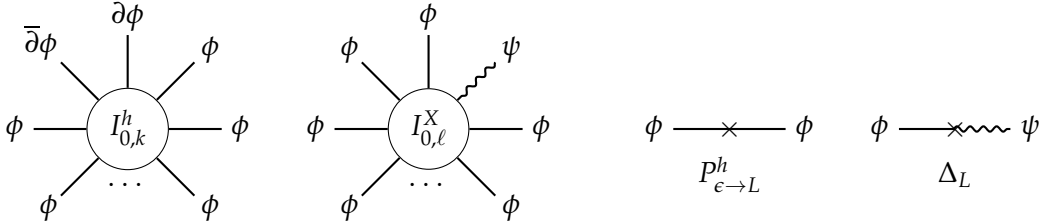
\vspace{1ex}

\begin{center}
        \begin{tikzpicture}
        
\node[circle] (h) at (-11,0) [draw] {$I^h_{0,k}$};
\node[circle] (x) at (-7,0) [draw] {$I^X_{0,\ell}$};
\draw[thick] (-11,1.25) -- (h.90);
\draw[thick] (-11.887,-0.887) --(h.225);
\draw[thick] (-10.113,-0.887) -- (h.-45);
\draw[thick] (-12.25,0)--(h.180);
\draw[thick] (-9.75,0)--(h.0);
\draw[thick] (-11.887,0.887)--(h.135);
\draw[thick] (-10.113,0.887)--(h.45);

\node at (-11,-0.75) {$\dotsb$};
\node at (-11,1.5) {$\partial \phi$};
\node at (-12.15, 1.15) {$\overline{\partial} \phi$};
\node at (-9.9,-1.05) {$\phi$};
\node at (-12.5,0) {$\phi$};
\node at (-12.05,-1.05) {$\phi$};
\node at (-9.5,0) {$\phi$};
\node at (-9.9, 1.05) {$\phi$};

\draw[thick] (-7,1.25) -- (x.90);
\draw[thick] (-7.887,-0.887) --(x.225);
\draw[thick] (-6.113,-0.887) -- (x.-45);
\draw[thick] (-8.25,0)--(x.180);
\draw[thick] (-5.75,0)--(x.0);
\draw[thick] (-7.887,0.887)--(x.135);
\draw[thick,decorate, decoration={snake, segment length=5,amplitude=1}] (-6.113,0.887)--(x.45);

\node at (-7,-0.75) {$\dotsb$};
\node at (-7,1.5) {$\phi$};
\node at (-8.05, 1.05) {$\phi$};
\node at (-5.9,-1.05) {$\phi$};
\node at (-8.5,0) {$\phi$};
\node at (-8.05,-1.05) {$\phi$};
\node at (-5.5,0) {$\phi$};
\node at (-5.9, 1.05) {$\psi$};

 \node (l1) at (-4,0) {$\phi$};
 \node (m1) at (-3,0) {$\times$};
 \node at (-3,-0.6) {$P_{\epsilon \to L}^h$};
\node (r1) at (-2,0) {$\phi$};

\draw[thick] (l1)--(r1);

 \node (l2) at (-1,0) {$\phi$};
 \node (m2) at (0,0) {$\times$};
 \node at (0,-0.6) {$\Delta_L$};
\node (r2) at (1,0) {$\psi$};

\draw[thick] (l2) -- (0,0);
\draw[thick,decorate, decoration={snake, segment length=5,amplitude=1}] (-0.03,0) --(r2);

        \end{tikzpicture}
\end{center}

\subsection{Reduction to $\beta_h$: a cohomological calculation}

We decompose the one-loop $\beta$ function into two pieces and show that only the component depending on the metric $h$ survives at the level of cohomology.

Note that in the proof of Proposition \ref{CME}, we showed that the two pieces of the action functional bracket to zero, i.e.,  $\{S_X , S_h \} = 0$.
In particular, the differentials $\{S_X, - \}$ and $\{S_h, -\}$ commute.
The following is an immediate corollary.

\begin{cor}
We can decompose $\beta^{(1)}$ as
\[
\beta^{(1)} = \beta^{(1)}_h + \beta^{(1)}_X \in H^{0} (\sO_{loc} (\sE)),
\]
where $\beta^{(1)}_h$ (resp. $\beta^{(1)}_X$) only involves diagrams with vertex type $I^h_{0,k}$ (resp. $I^X_{0,k}$).
\end{cor}

\begin{prop}\label{prop:betaisbetah}
The term $\beta^{(1)}_X$ vanishes. Consequently,
\[
\beta^{(1)} = \beta^{(1)}_h \in H^0 (\sO_{loc} (\sE)).
\]
\end{prop}

To prove the proposition, we will proceed via a spectral sequence argument, where we consider 
\[
(\sO_{loc} (\sE), \{S,-\} )=(\sO_{loc} (\sE), \{S_x,-\} + \{S_h, -\})
\] as a double complex.

\begin{lem}\label{lem:difflemma}
For the decomposition $\beta^{(1)} = \beta^{(1)}_h + \beta^{(1)}_X$, we have
\begin{itemize}
\item[(a)] $\left \{S_X , \beta^{(1)}_h \right \} = 0$;
\item[(b)] $\left \{S_h, \beta^{(1)}_h \right \} = \pm \left \{S_X, \beta^{(1)}_X \right \}$.
\end{itemize}
\end{lem}

\begin{proof}[Proof of Proposition \ref{prop:betaisbetah}]
Consider the spectral sequence of the relevant double complex with $E_1$ page given by $H^\ast (\sO_{loc} (\sE), \{S_X, -\})$. Via the preceding lemma (part (a)), we see that $\beta^{(1)}_h$ is closed and hence survives to the $E_1$ page.  Part (b) of the lemma further implies that $[\beta^{(1)}_h]_{E_1}$ is closed with respect to the differential and hence survives to the $E_2$ page. The $E_2$ page is concentrated on a single row, so the spectral sequence collapses. Hence, we have
\[
\beta^{(1)} = \left [\beta^{(1)}_h + \beta^{(1)}_X \right ]_{E_\infty} = \left [\beta^{(1)}_h \right ]_{E_2} = \beta^{(1)}_h.
\]
\end{proof}

\subsection{Computation of $\beta_h$.}

We now compute $\beta_h$ by reducing to a single Feynman weight computation and identifying the result in terms of geometric data on the target manifold $X$.

Recall that the one-loop quantization consisted of a family of functionals of the form $\{I[L] + \hbar J[L]\}$ where $I[L]$ was the naive quantization defined in terms of the weight expansion of the classical interaction via the propagators, and $J[L]$ was a quantum correction term we introduced to ensure the theory satisfied the QME.
Since we are working modulo $\hbar^2$, the quantity $\hbar J[L]$ is defined in terms of a Feynman expansion over genus zero graphs, that is, trees.
Hence, it has a well-defined $\epsilon \to 0$ limit and does not contribute to the divergences. 

Thus, the only part of the theory that contributes to the $\beta$-function is the na\"{i}ve quantization 
\ben
I[L] = \lim_{\epsilon \to 0} W(P_{\epsilon \to L}, I + I^{CT}(\epsilon)) .
\een

\begin{prop}
If $\Gamma$ is a wheel with at least two vertices then the weight $W_{\Gamma}(P_{\epsilon \to L}, I)$ has well-defined $\epsilon \to 0$ limit. 
That is, these graphs are UV finite, and hence do not contribute any counterterms.
\end{prop}

\begin{proof}
The result follows from a simple power counting argument. The singular behavior of the diagrams is the same as scalar field theory in dimension 2.
It is then a computation--completely analogous to those of Section \ref{sect:scalar}--that such a wheel with $k$ vertices has $\epsilon \to 0$ asymptotic behavior $\epsilon^{k-1} \log \epsilon + O(\epsilon^{k-1})$ and thus converges absolutely for $k \ge 2$.
\end{proof}

Hence, we are left to consider the weight of following type of one-loop diagrams.
The index $k$ indicates the number of external edges at the vertex.
We denote the $\epsilon$-dependent weight by $\cT_k (\epsilon)$. 

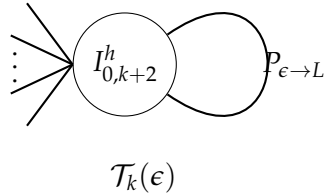
\vspace{1ex}

\begin{center}
        \begin{tikzpicture}
%

\node at (8.25,-1.5) {$\cT_k (\epsilon)$};
\node[circle] (c) at (8,0) [draw] {$I^h_{0,k+2}$};

 \node (uulc) at (6.6,0.95) {};
 \node (ulc) at (6.35, 0.45) {};
  \node (llc) at (6.35,-0.45) {};
  \node (lllc) at (6.6,-0.95) {};
\node[rotate=90] at (6.55,0) {$\dotsb$};
\tikzset{every loop/.style={min distance=10mm,in=-35,out=35,looseness=7}};
\path[thick] (c) edge  [loop above]  ();
\node at (10.25,0) {$P_{\epsilon \to L}$};

\draw[thick] (ulc)--(c.180);
\draw[thick] (llc)--(c.180);
\draw[thick] (uulc) --(c.180);
\draw[thick] (lllc) --(c.180);

        \end{tikzpicture}
\end{center}

\begin{lem}\label{lem:cohomologous}
As elements of $ (\sO(\sE)\otimes C^\infty ((\epsilon,L)), \{S,-\})$, the following are cohomologous
\[
\sum_{k} \cT_k(\epsilon) \sim \cT_2 (\epsilon) .
\]
\end{lem}

\begin{proof}
Recall the infinite jet map 
\[
j_\infty :  \Sym^2 (T_X^\vee) \hookrightarrow \Omega^*_X\left(\sJ \left(\Sym^2 (T_X^\vee)\right)\right) = \clie^*\left(\fg_X ; \Sym^2(T_X^\vee)\right),
\]
which is actually a quasi-isomorphism.  By construction (see Section \ref{sect:classical}), 
\[
\sum_k I^h_{0,k} = \int_\Sigma j_\infty (h) (\partial \phi , \overline{\partial} \phi ; e^\phi).
\]
Further, bracketing with the Lie theoretic part of the action, $S_X$, recovers the Chevalley-Eilenberg differential, i.e., $\{S_X, -\} = d_{\rm Lie}$. Finally, as $\{S_X, -\}$ and $\{S_h, -\}$ commute, the proposition follows from the fact that $j_\infty$ is a quasi-isomorphism.
\end{proof}

%
%
%
%

Next, it is a straightforward computation with scalar heat kernels to see that $\cT_2 (\epsilon)$ is actually purely $\log \epsilon$ divergent.

\begin{lem}\label{lem:logdiv}
The Feynman weight $\cT_2(\epsilon) $ is purely $\log \epsilon$ divergent, so
\[
\cT_2(\epsilon) = \log(\epsilon) \cT_2
\]
for some local functional $\cT_2$.
\end{lem}
\begin{proof}
The analytic part of the propagator is of the form $\int_{t=\epsilon}^L \frac{1}{4 \pi t} e^{-|x-y|^2/4t} \d t$.
Thus, the analytic contribution to the weight of the tadpole diagram (where we set $x = y$) is purely log $\epsilon$ divergent.
\end{proof}

We can summarize the preceding arguments as follows:
\[
\beta^{(1)} = \left [\sO^{(1)}_\beta \right ] = [\mathrm{sing}_{\log \epsilon} \cT_2 (\epsilon)] = [\cT_2] \in H^0 (\sO_{loc}(\sE)) .
\]

We would now like to realize the $\beta$-function in terms of the geometry of the target Riemannian manifold $(X,h)$. To begin, following Section \ref{sect:betafunc}, we could describe a basis for $H^0 (\sO_{loc} (\sE))$. Alternatively, we recall the map
\[
\mathfrak{R} : \Sym^2 (T_X^\vee) \to \sO_{loc} (\sE)
\]
we constructed in Section \ref{sect:obsdef}.

\begin{thm}\label{thm:main}
As elements of $H^0 (\sO_{loc} (\sE))$, we have an equality
\[
\beta^{(1)} = -\frac{1}{12 \pi} \left [ \mathfrak{R} (\mathrm{Ric}) \right ].
\]
\end{thm}

\begin{proof}

First note that in the $\epsilon \to 0$ limit, $P_{\epsilon \to L}$ becomes a delta function on the diagonal: that is, we take a trace over pairs of jet indices.  Next, up to a constant, the covariant curvature tensor $R_{ikj\ell}$ corresponding to the metric $h$ determines the $4$-valent vertex $I^h_{0,4}$; for $k>4$, the $I^h_{0,k}$ are just the jet expansions of this tensor. Combining the previous two observations, we have an equivalence of local functionals

\[
\cT_2 = -\frac{1}{4 \pi} \int_\Sigma j_2 (h) (\partial \phi , \overline{\partial} \phi; \phi_k , \phi_l)  \delta_{kl}
= -\frac{1}{4 \pi}  \int_\Sigma \left (\frac{1}{3}\right )  R_{ikj \ell} (\partial \phi, \phi, \overline{\partial} \phi,\phi) \delta_{k \ell}
= -\frac{1}{12 \pi} \int_\Sigma R_{ikjk} (\partial \phi, \phi, \overline{\partial} \phi,\phi).
\]
Thus we have
\[
\cT_2  = - \frac{1}{12 \pi}  \int_\Sigma \mathrm{Ric} (\partial \phi , \overline{\partial} \phi).
\]
The same argument as in Lemma \ref{lem:cohomologous} proves that
\[
-\frac{1}{12 \pi} \mathfrak{R} (\mathrm{Ric}) \quad \text{ and }  \quad - \frac{1}{12 \pi}  \int_\Sigma \mathrm{Ric} (\partial \phi , \overline{\partial} \phi)
\]
are cohomologous as elements of $\sO_{loc} (\sE)$.
\end{proof}

Using the image of the map $\mathfrak{R} : \Sym^2 (T_X^\vee) \to \sO_{loc} (\sE)$, we can rephrase the theorem as follows.

\begin{cor}\label{prop:betaricci}
With respect to the renormalization scheme where  $\epsilon^{-1}$ and $\log \epsilon$ are purely singular, we have
\[
\beta^{(1)} (h) = -\frac{1}{12 \pi} \mathrm{Ric} (h) \in \Sym^2 (T_X^\vee) \subset H^0 (\sO_{loc} (\sE)).
\]
\end{cor}

\begin{rmk}
One could give an alternative proof of Theorem \ref{thm:main} as follows.  In Section \ref{sect:obsdef}, we have identified $H^0(\sO_{loc} (\sE))$ as a direct sum.  Next, argue locally on the target: locally the $\L8$ algebra $\fg_X$ is curved, but abelian, so there are no interesting vertices of type $I_X$.  Consequently, we know that $\beta^{(1)}$ is a tensor on the target. Finally, we choose geodesic normal coordinates and compute this tensor at a point in $X$.
\end{rmk}

\appendix
\section{Homotopy renormalization group flow}\label{app:app}

The homotopy renormalization group flow equation can be described in terms of Feynman graphs. Note that our description is for an arbitrary functional on a space of fields $\sE$. Further, we will work relative to an arbitrary dg algebra $\cA$ equipped with a nilpotent ideal $\cI$.

\begin{dfn}
A graph $\cG$ consists of the following data:
\begin{enumerate}
 \item A finite set of vertices $V(\cG)$;
 \item A finite set of half-edges $H(\cG)$;
 \item An involution $\sigma: H(\cG)\rightarrow H(\cG)$. The set of fixed points of this map is denoted by $T(\cG)$ and is
called the set of tails of $\cG$. The set of two-element orbits is denoted by $E(\cG)$ and is called the set of internal edges of
$\cG$;
 \item A map $\pi:H(\cG)\rightarrow V(\cG)$ sending a half-edge to the vertex to which it is attached;
 \item A map $g:V(\cG)\rightarrow \mathbb{Z}_{\geqslant 0}$ assigning a genus to each vertex.
\end{enumerate}
\end{dfn}
It is clear how to construct a topological space $|\cG|$ from the above abstract data. A graph $\cG$ is called $connected$ if
$|\cG|$ is connected. The genus of the graph $\cG$ is defined to be 
\[
g(\cG):=b_1(|\cG|)+\sum_{v\in V(\cG)}g(v),
\] 
where
$b_1(|\cG|)$ denotes the first Betti number of $|\cG|$. Let 
\[
\sO^+(\sE)\subset \sO(\sE)[[\hbar]]
\] 
be the subspace consisting of those
functionals  which are at least cubic modulo $\hbar$ and the nilpotent ideal $\mathcal{I}$ in the base ring
$\mathcal{A}$.  Let $F\in \sO^+(\sE)$ be a functional,  which can be expanded as
\[
F=\sum_{g,k\geq 0}\hbar^g F_{g}^{(k)}, \quad F_{g}^{(k)}\in \sO^{(k)}(\sE).
\]
We view each $F_{g}^{(k)}$ as an
$S_k$-invariant linear map
\[
F_{g}^{(k)}: \mathcal{E}^{\otimes k}\rightarrow\mathcal{A}.
\]
With the propagator $P_{\epsilon \to L}$, we will describe the {\it (Feynman) graph weights}
\[
W_\cG(P_{\epsilon \to L},F)\in \sO^+(\sE)
\] 
for any connected graph $\cG$. We label each vertex $v$ in $\cG$ of genus $g(v)$ and valency $k$ by
$F^{(k)}_{g(v)}$. This defines an assignment
\[
F(v):\mathcal{E}^{\otimes H(v)}\rightarrow \cA,
\]
where $H(v)$ is the set of half-edges of $\cG$ which are incident to $v$.
Next, we label each internal edge $e$ by the propagator 
\[
P_e=P_{\epsilon \to L}\in\mathcal{E}^{\otimes H(e)},
\]
where $H(e)\subset H(\cG)$ is the two-element set consisting of the half-edges forming $e$. We can then contract
\[
\otimes_{v\in V(\cG)}F(v): \mathcal{E}^{H(\cG)}\rightarrow \cA
\]
with 
\[
\otimes_{e\in E(\cG)} P_e\in\mathcal{E}^{H(\cG)\setminus T(\cG)}
\] 
to yield a linear map
\[
W_\cG(P_{\epsilon \to L},F) : \mathcal{E}^{\otimes T(\cG)}\rightarrow \cA.
\]

\begin{dfn}
We define the (homotopy) RG flow operator with respect to the propagator $P_{\epsilon \to L}$ 
\[
   W(P_{\epsilon \to L}, -): \sO^+(\sE)\to \sO^+(\sE), 
\]
by
\begin{equation}\label{RG-flow}
W(P_{\epsilon \to L}, F):=\sum_{\cG}\frac{\hbar^{g(\cG)}}{\lvert \text{Aut}(\cG)\rvert}W_\cG(P_{\epsilon \to L}, F)
\end{equation}
where the sum is over all connected graphs.
\end{dfn}

Equivalently, it is useful to describe the (homotopy) RG flow operator formally via the simple equation 
\[
e^{W(P_{\epsilon \to L}, F)/\hbar}=e^{\hbar \partial_{P_{\epsilon \to L}}} e^{F/\hbar}.
\]

\begin{dfn} A family of functionals $F[L] \in \sO^+(\sE)$ parametrized by $L>0$ is said to satisfy the homotopy renormalization group flow equation (hRGE) if for each $0 < \epsilon < L$
\[
    F[L]=W(P_{\epsilon \to L}, F[\epsilon]).
\]
\end{dfn}


\bibliographystyle{amsplain}
\bibliography{beta}

\end{document}